\documentclass[a4paper,11pt]{article}
\usepackage{geometry}
\usepackage{mathtools,amsfonts,amsthm,mathrsfs,amssymb}
\usepackage{amsmath}
\usepackage[utf8]{inputenc}
\usepackage[T1]{fontenc} 
\usepackage{lmodern} 
\usepackage{bm}
\usepackage{enumitem}
\usepackage{todonotes}
\usepackage{array}
\usepackage{tikz}
\usepackage{xspace}
\usepackage{cite}
\usepackage{hyperref}
\usepackage{cleveref}
\usepackage{authblk}
\usepackage{comment}

\newtheorem{thm}{Theorem}
\newtheorem{lemma}[thm]{Lemma}

\newtheorem{proposition}[thm]{Proposition}
\newtheorem{corollary}[thm]{Corollary}

\newtheorem{claim}{Claim}[thm]

\theoremstyle{definition}
\newtheorem{problem}[thm]{Problem}
\newtheorem{remark}[thm]{Remark}

\newtheorem*{definition*}{Definition}

\def\bigO{\ensuremath{{\cal O}}\xspace}

\newcommand{\dI}{\vec{\mathcal{I}}}
\newcommand{\cI}{\mathcal{I}}

\newcommand{\wI}{\widetilde{I}}
\newcommand{\tin}{\texttt{tin}}
\newcommand{\tout}{\texttt{tout}}

\newcommand{\sset}[2]{\left\{#1 \;\middle |\; #2 \right \}}

\newcommand{\ceil}[1]{\left\lceil #1 \right\rceil}

\newcommand{\restrict}[2]{{#1}_{|{#2}}}

\newenvironment{subproof}{\par\noindent {\it Proof of claim}.\ }{\hfill$\lozenge$\par\vspace{11pt}}

\DeclareMathOperator{\lca}{LCA}

\newcommand{\ext}{\mathscr{C}_0}

\newcommand{\DC}[1]{{\color{olive} [{\bf David}: #1]}\xspace}

\newcommand{\FP}[1]{{\color{orange}{{\bf FP}: #1}}\xspace}

\title{Colouring the interference digraph of a set of requests in a bidirected tree}

\author[1]{Hugo Boulier}
\author[2]{David Coudert}
\author[2]{Fr\'ed\'eric Havet}
\author[3]{Fran\c{c}ois Pirot}

\affil[1]{\'Ecole normale sup\'erieure de Rennes, France}
\affil[2]{Universit\'e C\^ote d'Azur, CNRS, Inria, I3S, Sophia Antipolis, France}
\affil[3]{LISN, Université Paris-Saclay, Gif sur Yvette, France.}

\begin{document}

\maketitle

\begin{abstract}
In this paper, we investigate the impact of the broadcast effect arising in filterless optical networks on the computational complexity of the wavelength assignment problem.
We model conflicts using an appropriate interference digraph, whose proper colourings correspond to feasible wavelength assignments. Minimizing the number of required wavelengths therefore amounts to determining the chromatic number of this interference digraph.
Within this framework, we first present a polynomial-time 2-approximation algorithm for minimizing the number of wavelengths. We then show that the problem is fixed-parameter tractable when parameterized by the number $k$ of available wavelengths.
We also derive polynomial-time algorithms for computing the independence and clique numbers of this interference digraph.
\end{abstract}

\section{Introduction}

In this paper, we investigate the computational complexity of the wavelength assignment problem in filterless optical networks. These networks use wavelength division multiplexing (WDM) technology to share the optical spectrum of optical fibers and passive optical devices (splitters and combiners) to perform the interconnections at the router nodes.
Recall that an optical network is constituted of a set of router nodes connected by directed optical fibers, and so the network topology of such networks can naturally be modelled as a directed graph.
In WDM networks, each traffic request must be assigned a \emph{lightpath} from its source to its destination, that is a path in the network and a wavelength. The wavelength assignment must satisfy the WDM constraint that two lightpaths using a same optical fiber must be assigned different wavelengths.
The optical router nodes are usually equipped with Reconfigurable Optical Add-Drop Multiplexers (ROADMs) to guide each lightpath from an input port to an output port.  The use of ROADMs in WDM networks allows for a very efficient use of the optical spectrum but the overall cost of the network is high.
In order to find a better compromise between the cost of the network and the spectrum efficiency, the use of cheap passive optical components has been proposed~\cite{tza99}. 

In filterless optical networks~\cite{sav10,Ayoub2022}, an optical signal entering a router node is simply broadcast to all fibers outgoing the node, except the one going in the opposite direction of the input optical fiber. This is done using cheap passive optical devices (splitters and combiners). Furthermore, some nodes are equipped with blockers to stop the propagation of the optical signal at certain nodes. To simplify the management of such networks, the blockers are usually placed in order to separate the network into a set of trees, and a traffic request is assigned to a unique tree. Due to the broadcast effect, the spectrum efficiency of filterless optical networks is reduced compared to WDM networks equipped with ROADMs. Indeed, the wavelength used for a lightpath is broadcast into an entire tree, thus increasing the number of conflicts for the assignment of wavelengths. However, both the capital expenditure of building the network and the operational expenditure of controlling the network operations are reduced.

The general optimization problem in filterless optical network is, given a network topology and a set of traffic requests, to partition the network topology into a set of trees and assign each traffic request a lightpath in a tree in such a way that the total number of distinct wavelengths used in the network is minimized. Several exact and heuristic methods have been proposed for solving this optimization problem~\cite{tre07,tre13,tre17,jau18,jau21_JLT,Jaumard2021,lv2023,Etezadi2024,Karandin2025}. 
In particular, \cite{jau21_JLT,Jaumard2021} proposed a mathematical model of the problem that considers as variables sub-networks with a tree topology along with traffic requests and lightpaths. Since routing is obvious in a tree, this model raises the question of the computational complexity of the wavelength assignment problem that we study in this paper. 

Observe that when a path is given as input for each traffic request, the wavelength assignment problem in WDM networks equipped with ROADMs is equivalent to the problem of colouring a set of paths in a directed graph so that two paths sharing an arc get different colours~\cite{Chlamtac1992}. This later problem is NP-complete~\cite{ERLEBACH2001,Beauquier1997} even when the digraph is a bidirected cycle~\cite{Tucker1975} or a bidirected tree~\cite{GOLUMBIC1985,Tarjan1985}. Indeed, the paths colouring problem corresponds to a vertex colouring problem in the conflict graph of the paths, that is a graph with one vertex per path and an edge between two vertices if the corresponding paths share an arc.

\subsection{Modelling}

In filterless optical networks, the broadcast effect increases the number of conflicts between traffic requests for the wavelength assignment. We model these conflicts using an \emph{interference digraph} in which there is a vertex for each traffic request, and there is an arc from vertex $r$ to vertex $r'$ if the broadcast tree of request $r$ intersects the lightpath of request $r'$. Then, assigning wavelengths to requests is equivalent to properly colouring the vertices of the interference digraph. 

More precisely, the broadcast tree can be modelled by a \emph{bidirected tree} $T$, that is a digraph obtained from an undirected tree by replacing each edge by two arcs in opposite directions.
A {\it request} in $T$ is a directed subpath of $T$ of length at least $1$.
Let $r$ be a request or more generally a directed path.  We denote by $s_r$ its first vertex, by $t_r$ its ultimate vertex,
by $s_r^+$ its second vertex and by $t_r^-$ its penultimate vertex. The \emph{emission arc} of $r$ is the arc $e^+_r=(s_r,s_r^+)$, and the \emph{reception arc} of $r$ is the arc $e^-_r=(t_r^-,t_r)$.

Let $R$ be a set of requests in a bidirected tree $T$.
A request $r$ {\it interferes on} a request $r'$ if and only if $T[s_r, t_{r'}]$, the unique directed path from $s_r$ to $t_{r'}$ in $T$, has first arc $e^+_r$ and ultimate arc $e^-_{r'}$.
The \emph{interference digraph} of $R$ in $T$, denoted by $\dI(R,T)$, is the digraph with vertex set $R$ in which $(r,r')$ is an arc if and only if $r$ interferes on $r'$.
Two requests {\it interfere with} each other if one of them interferes on the other.
The \emph{interference graph} of $R$ in $T$, denoted by $\cI(R,T)$, is the underlying graph of $\dI(R,T)$. In other words, it is the graph with vertex set $R$ in which $\{r,r'\}$ is an edge if and only if $r$ and $r'$ interfere.

\subsection{Our results}
In this paper, given a bidirected tree $T$ and a set of requests $R$ in $T$, we are interested in $\omega(R,T)$ the size of a largest subset of pairwise interfering requests and $\alpha(R,T)$ the size of a largest subset of pairwise non-interfering (or independent) requests. Equivalently, $\omega(R,T)$ is the clique number of $\cI(R,T)$, that is $\omega(R,T) = \omega(\cI(R,T))$, and $\alpha(R,T)$ is the independence number of $\cI(R,T)$, that is $\alpha(R,T) = \alpha(\cI(R,T))$.
We are also interested in the minimum number $\chi(R,T)$ of colours (wavelengths or frequencies) to assign to the requests so that any two interfering requests are assigned different colours. This is the chromatic number of $\cI(R,T)$, that is $\chi(R,T) = \chi(\cI(R,T))$.

We show in Section~\ref{sec:inde-clique}, that given $R$ and $T$, one can compute $\alpha(R,T)$ and $\omega(R,T)$ in polynomial time.
Then, in Section~\ref{sec:chi-bound}, we prove $\chi(R,T)\leq 2\, \omega(R,T)$, and describe a simple $2$-approximation polynomial-time algorithm for $\chi(R,T)$.

We then study in Section~\ref{sec:kcolourability} the complexity of computing $\chi(R,T)$ and the following associated decision problem.

\noindent\textsc{Interference Colourability}\\
\underline{Input:} A set $R$ of requests in a bidirected tree $T$, and an integer $k$.\\
\underline{Question:} $\chi(R,T) \leq k$ ?

\medskip
In particular, we consider the restrictions of this problem when $k$ is fixed.

\noindent\textsc{Interference $k$-Colourability}\\
\underline{Input:} A set $R$ of requests in a bidirected tree $T$.\\
\underline{Question:} $\chi(R,T) \leq k$ ?

\medskip
We show in Section~\ref{sec:3colourability} that \textsc{Interference $3$-Colourability} can be solved in $\bigO(|R|^2)$ time. Then, in Section~\ref{sec:4colourability}, we prove that \textsc{Interference $k$-Colourability} with $k \geq 4$ is fixed parameter tractable (FPT) when parameterized by $k$, or equivalently by the size of the largest clique in $\cI(R,T)$. We propose an FPT algorithm with time complexity in $\bigO(k36^k|R|^3)$
for solving this problem.

We conclude this paper in Section~\ref{sec:conclusion} with some open problems.

\section{Preliminaries}

\subsection{Notations and definitions}

An {\it out-tree} (resp. {\it in-tree})  is an oriented rooted tree in which all the arcs are directed away from  (resp. towards) the root. 
Given a bidirected tree $T$ and a vertex $z\in V(T)$, we denote by $T[z]$ the rooted tree isomorphic to $T$ whose root is $z$. A \emph{bough} in a bidirected tree is a bidirected path between two leaves, and a \emph{branch} in a rooted bidirected tree is a bidirected path from the root to any leaf. 
An {\it in-branch} in a rooted bidirected tree is a directed path from a leaf to the root, and an {\it out-branch} in a rooted bidirected tree is a directed path from the root to the leaf.

Given a directed path $P=(v_0, \dots, v_\ell)$, we denote by $P[v_i,v_j]$ the directed subpath $(v_i, \dots , v_j)$, for any $0\le i \leq j \le \ell$.
More generally, if $T$ is a bidirected tree or an in-tree or an out-tree, we denote by $T[x,y]$ the unique directed path from $x$ to $y$ in $T$ (if it exists).
If $T$ is a bidirected tree, we denote by 
$T\langle x,y\rangle$ the bidirected path with ends $x$ and $y$ in $T$.

The {\it converse} of a digraph $D$, denoted by $\overleftarrow{D}$, is the digraph obtained from $D$ by reversing the direction of all arcs: $V(\overleftarrow{D}) = V(D)$ and $A(\overleftarrow{D}) = \{(u,v) : (v,u) \in A(D)\}$.

Given a graph $G$ and a set $S$ of vertices of $G$, we denote by $\overline{S}$ the set $V(G)\setminus S$, and by $G\langle S\rangle$ the subgraph of $G$ induced by $S$. 
The neighbourhood of a vertex $v$ is denoted by $N(v)$ and its {\it closed neighbourhood} is $N[v]=N(v) \cup \{v\}$.

\subsection{Reducing the problem}\label{subsec:reducing}

\paragraph{Reducing the size of the tree.}

Let $R$ be a set of requests in a bidirected tree $T$.
Observe that if two opposite arcs in $T$ are neither emission nor reception arc, then we can contract those two arcs (i.e. remove them and identify their end-vertices) in both $T$ and the requests to get a new bidirected tree $T'$ and set of requests $R'$ with same interference digraph: $\dI(R',T') = \dI(R,T)$.
Doing so, contracting one after another the pair of opposite arcs with no emission and no reception arcs, we can reduce the problem in $\bigO(|T|)$ time to an instance $(T^*, R^*)$ 
such that $\dI(R^*,T^*) = \dI(R,T)$ and $|V(T^*)|\leq 2 |R^*| +2$.

Hence, free to first perform the above $\bigO(|T|)$-time reduction, 
{\bf we assume that every pair of opposite arcs contains an emission or reception arc}, and so  $$|V(T)|\leq 2 |R| +2.$$

\paragraph{Computing the interference digraph.}

For each of the $|R|$ request $r\in R^*$, one can build in $\bigO(|T|)= \bigO(|R|)$ time the out-tree $T^+(r)$ of vertices that are reachable from $s^+_r$ in $T\setminus \{s_rs^+_r, s^+_rs_r\}$ and thus find the out-neighbours of $r$ in  $\dI(R,T)$. 
Hence in $\bigO(|R|^2)$ time one can compute $\dI(R,T)$.
Note that this time is fastest possible in general because $\dI(R,T)$ can have 
$\Omega(|R|^2)$ arcs.

In the algorithms we provide, it is sometimes needed to construct interference digraphs of subsets of requests on a tree. Whenever this is required, the associated cost never exceeds the cost of the operations performed on that digraph. For that reason, and to lighten the proofs, we make these constructions implicit and never mention their cost.


\subsection{Structural properties}\label{sec:structuralproperties}

Let $T$ be a bidirected tree rooted at $z$. 
It is the arc-union of the in-tree $T^-$, whose arcs are those directed towards the root, and the out-tree $T^+$, whose arcs are those directed away from the root.

In $T$, a directed path (and in particular a request) is either {\it converging} if it is directed towards $z$ (i.e. contained in $T^-$), either {\it diverging} if it is directed away from $z$ (i.e. contained in $T^+$),
and {\it unimodal} otherwise (i.e. the concatenation of a subpath of $T^-$ and a subpath of $T^+$).
Given a set $R$ of requests on $T$, we denote by $R^-$ (resp. $R^+$, $R^\vee$) the set of converging (resp. diverging, unimodal) requests of $R$. 
By definition, $(R^-, R^+, R^\vee)$ is a partition of $R$.

If $r$ is a directed subpath of $T$, the {\it middle} of $r$, denoted by $m_r$, is its vertex which is closest to the root. If $r$ is a converging (resp. diverging) path, then $m_r=t_r$ (resp. $m_r=s_r$). If $r$ is unimodal, then it is the vertex of $r$ which is the head of a converging arc and the tail of a diverging arc. The predecessor of $m_r$ along $r$ is denoted by $m^-_r$ and its successor by $m^+_r$, should they exist.

Let $x,y$ be two vertices of $T$. We say that $x$ is an {\it ancestor} of $y$ and $y$ is a {\it descendant} of $x$ if $T[x,y]$ is diverging. A path of one vertex is diverging, so every vertex is an ancestor of itself.
We say that $x$ and $y$ are \emph{related} if $x$ is an ancestor of $y$ or $y$ is an ancestor of $x$. In other words, $x$ and $y$ are related if and only if $T[x,y]$ is either converging or diverging.
Observe that two ancestors of a same vertex are related.

The {\it least common ancestor} of $x$ and $y$ is the vertex that is furthest away from the root $z$ which is both an ancestor of $x$ and of $y$. This corresponds to the middle of $T[x,y]$.



\medskip

We can now state and prove a lemma that characterizes when two requests interfere. We will sometimes use this lemma without 
explicitly referring to it.

\begin{lemma}
\label{lem:interference}
Let $T$ be a bidirected tree rooted at $z$.
\begin{enumerate}[label={\rm(\roman*)}]
    \item Two converging requests $r,r'$ interfere if and only if $t_r^-$ and $t_{r'}^-$ are related.
    \label{it:converging}
    \item Two diverging requests $r,r'$  interfere if and only if $s_r^+$ and $s_{r'}^+$ are related.
    \label{it:diverging}

    \item A converging request $r$ and a diverging  request $r'$ interfere if and only if $s_r$ and $t_{r'}$ are not related. 
    \label{it:mixte}

    \item Two unimodal requests $r,r'$ do not interfere if and only if $m_{r}=m_{r'}$, $s_{r}$ and $t_{r'}$ are related, and $s_{r'}$ and $t_{r}$ are related. \label{it:unimodal}
    
    \item A unimodal request $r$ and a converging request $r'$ do not interfere if and only if $m_r$ is an ancestor of $t_{r'}$ and
    $t_r$ and $s_{r'}$ are related.
    \label{it:uni-conv}
    
    \item A unimodal request $r$ and a diverging  request $r'$ do not interfere if and only if $m_r$ is an ancestor of $s_{r'}$ and
    $s_r$ and $t_{r'}$ are related.
    \label{it:uni-div} 
    
\end{enumerate}
\end{lemma}

\begin{proof}

\ref{it:converging} Let $r$ and $r'$ be two converging requests. 

Assume $r$ interferes on $r'$. 
Then $T[s_r, t_{r'}]$ has first arc $e^+_r$ and ultimate arc $e^-_{r'}$. Since $r$ and $r'$ are converging, the two arcs $e^+_r$ and $e^-_{r'}$ are directed towards the root, so $T[s_r, t_{r'}]$ is a converging path.
Hence $t^-_r$ and $t^-_{r'}$ are both ancestors of $s_r$ and thus they are related.

Conversely assume that $t^-_r$ and $t^-_{r'}$ are related.
Without loss of generality, we may assume that $t^-_{r'}$ is an ancestor of $t^-_{r}$. But then $T[s_r, t_{r'}]$ is a converging path with first arc $e^+_r$ and ultimate arc $e^-_{r'}$. Hence $r$ interferes on $r'$.

\medskip

\ref{it:diverging} Similar to \ref{it:converging} by directional duality.

\medskip

\ref{it:mixte} Assume that $r$ and $r'$ interfere.
Since a directed path in $T$ is either converging, diverging or unimodal, it cannot have its first arc diverging and its ultimate arc converging.  In particular, there is no directed path in $T$ with first arc $e^+_{r'}$ and ultimate arc $e^-_{r'}$.
Hence $r'$ cannot interfere on $r$.
Thus $r$ interferes on $r'$, and $T[s_r, t_{r'}]$ is a directed path with first arc $e^+_r$ and ultimate arc $e^-_{r'}$.
Since $e^+_r$ is converging and $e^-_{r'}$ is diverging, $T[s_r,t_{r'}]$ is unimodal, and $s_r$ and $t_{r'}$ are not related.

Conversely, assume that $s_r$ and $t_{r'}$ are not related. Then 
$T[s_r, t_{r'}]$ is unimodal. Therefore, the first arc of this path is converging, and thus, it must be $e^+_r$ and its ultimate arc is diverging and thus must be $e^-_{r'}$.
Hence $r$ interferes on $r'$.

\medskip
\ref{it:unimodal}
Let $r$ and $r'$ be two unimodal requests.

Assume first that $s_r$ and $t_{r'}$ are not related, then $T[s_r, t_r']$ is a directed path going through their least common ancestor and with first arc $e^+_{r}$ and ultimate arc $e^-_{r'}$. 
Hence $r$ interferes on $r'$.

Similarly, if $s_{r'}$ and $t_r$ are not related, then $r'$ interferes on $r$.

Assume now that $s_r$ and $t_{r'}$ are related, and 
$s_{r'}$ and $t_r$ are related. Then necessarily $m_r=m_{r'}$. Moreover, one easily sees that $r$ and $r'$ do not interfere.

\medskip
\ref{it:uni-conv}
Let $r$ be a unimodal request and $r'$ a converging request.
If $m_r$ and $t_{r'}$ are not related then by \ref{it:mixte} $T[m_r, t_r]$ and $r'$ interfere, and thus $r$ and $r'$ interfere.

Henceforth, we may assume that $s_r$ and $t_{r'}$ are related.
Then $T[s_{r'}, t_{r}]$ is either a converging path or diverging path, and thus does not contain either the converging arc $e^+_{r'}$ or the diverging arc $e^-_{r}$. Thus  $r'$ does not interfere on $r$.

If $t_{r'}$ is an ancestor of $m_r$ and $t_{r'}\neq m_r$, then 
$T[s_r, t_{r'}]$ is a directed path with first arc $e^+_r$ and ultimate arc $e^-_r$. Therefore $r$ interferes on $r'$.
If $m_r$ is an ancestor of $t_{r'}$, then $T[s_r, t_{r'}]$ is a unimodal path and thus does not contain the converging arc $e^-_{r'}$.
Hence $r$ and $r'$ do not interfere.

\ref{it:uni-div} Similar to \ref{it:uni-conv} by directional duality.
\end{proof}

Since any set of pairwise related nodes in a tree must lie on a common branch, we immediately derive from Lemma~\ref{lem:interference}\ref{it:converging} and Lemma~\ref{lem:interference}\ref{it:diverging} the following.

\begin{corollary}\label{lem:clique-branch}
Let $T$ be a bidirected tree rooted at $z$.

\begin{enumerate}[label={\rm(\roman*)}]
\item If $W$ is a clique of converging requests, then there is an in-branch of $T$ containing the reception arcs of all requests in $W$.
\item If $W$ is a clique of diverging requests, then there is an out-branch of $T$ containing the emission arcs of all requests in $W$.
\end{enumerate}

\end{corollary}

A \emph{comparability graph} is a graph for which there exists a partial order such that two vertices are linked by an edge if and only if they are comparable in the partial order.  
A \emph{cobipartite graph} is the complement of a bipartite graph.
In other words, a cobipartite is a graph whose vertex set can be partitioned into two cliques.

Comparability and cobipartite graphs are perfect~\cite{Golumbic1980Ch5}. 
Hence $\omega(G) = \chi(G)$ for every comparability or cobipartite graph $G$.
Moreover the clique number, the independence number and the chromatic number can be computed in polynomial time for comparability and cobipartite graphs. 

Let $G_i$, $i\in I$, be graphs.
The \emph{disjoint union} of the $G_i$, $i\in I$, denoted by $\biguplus_{i\in I} G_i$, is the union of vertex-disjoint copies of each $G_i$.
The \emph{join} of the $G_i$, $i\in I$, denoted by $\bigoplus_{i\in I} G_i$, is the graph obtained from $\biguplus_{i\in I} G_i$ by adding all possible edges between (the copy of) $G_i$ and (the copy of) $G_j$ for all $i,j\in I, i\neq j$.

The following proposition is easy and well-known.
\begin{proposition}\label{prop:comp}
\begin{enumerate}[label={\rm(\roman*)}]
    \item The disjoint union and the join of comparability graphs are comparability graphs.
    \item  The join of cobipartite graphs is a cobipartite graphs.
\end{enumerate}
\end{proposition}

\begin{corollary}\label{cor:conv-comp}
Let $R$ be a set of requests on a rooted bidirected tree $T$. We have:
\begin{enumerate}[label={\rm(\roman*)}]
 \item $\cI(R^-,T)$ is a comparability graph, so $\omega(R^-,T)=\chi(R^-,T)$. \label{it:conv-comp:comparability-minus}
 \item $\cI(R^+,T)$ is a comparability graph, so $\omega(R^+,T)=\chi(R^+,T)$. \label{it:conv-comp:comparability-plus}
 \item $\cI(R^\vee,T)$ is a cobipartite graph, so $\omega(R^\vee,T)=\chi(R^\vee,T)$. \label{it:conv-comp:cobipartite}
\end{enumerate}
\end{corollary}
\begin{proof}
\ref{it:conv-comp:comparability-minus} follows directly from Lemma~\ref{lem:interference}\ref{it:converging}.

\medskip

\ref{it:conv-comp:comparability-plus} follows directly from Lemma~\ref{lem:interference}\ref{it:diverging}.

\medskip

\ref{it:conv-comp:cobipartite} Let $z$ be the root of $T$.

Let $x$ be a vertex of $T$. We denote by $R_x^\vee$ the set of unimodal requests of $R$ with middle $x$.
Let $S=\{s_1, \dots , s_p\}$ be the set of {\it sons} of $x$, that are the neighbours of $x$ which are not in $T[z,x]$.
Observe that for every request in $R_x^\vee$, $m^-_r$ and $m^+_r$ are distinct vertices in $S$. So let $i^-_r$ be the index such that $m^-_r=s_{i^-_r}$ and let $i^+_r$ be the index such that $m^+_r=s_{i^+_r}$.

Let $A_x$ (resp. $B_x$) be the set of requests of $R_x^\vee$ such that $i^-_r< i^+_r$ (resp. $i^-_r > i^+_r$). Clearly, $(A_x,B_x)$ is a partition of $R^\vee_x$. Moreover, by Lemma~\ref{lem:interference}\ref{it:unimodal}, $A_x$ and $B_x$ are cliques in $\cI(R^\vee_x,T)$.
Hence $\cI(R^\vee_x,T)$ is a cobipartite graph.

Now, by Lemma~\ref{lem:interference}\ref{it:unimodal},
$\cI(R^\vee,T)$ is the join of the  $\cI(R^\vee_x,T)$, $x\in V(T)$.
Thus, by Proposition~\ref{prop:comp}, $\cI(R^\vee,T)$ is a cobipartite graph.
\end{proof}

\begin{lemma}
\label{lem:nospan}
Let $r^-$, $r^+$, and $ r^\vee$ be a converging, a diverging, and a unimodal request, respectively, on a rooted bidirected tree $T$.
Then two requests of  $\{r^-,r^+,r^{\vee}\}$ interfere.
\end{lemma}

\begin{proof}
If $r^\vee$ interferes with none of $r^-, r^+$, then by Lemma~\ref{lem:interference}~\ref{it:uni-conv} and~\ref{it:uni-div}, $m_{r^\vee}$ is an ancestor of both $t_{r^-}$ and $s_{r^+}$, $t_{r^\vee}$ and $s_{r^-}$ are related, and $s_{r^\vee}$ and $t_{r^+}$ are related.

Then $T[s_{r^-}, t_{r^+}]$ is a unimodal path with middle $m^r$, first arc $e^+_{r^-}$, and ultimate arc $e^-_{r^+}$.
Hence $r^-$ interferes on $r^+$.
\end{proof}

\begin{lemma}
\label{lem:colour-size-1}
Let $R$ be a set of requests on a rooted bidirected tree $T$, and $I \subseteq R$ a set of independent requests. Then, either $I$ contains exactly two unimodal requests, or $|I\cap (R^+\cup R^\vee)|\le 1$, or $|I\cap (R^-\cup R^\vee)|\le 1$.
\end{lemma}

\begin{proof}
If $I$ contains two unimodal requests $r,r'$, then by Lemma~\ref{lem:interference}\ref{it:unimodal}, $r$ and $r'$ have arcs in opposite direction (e.g. $(m_r,m^+_r)$ and $(m^-_{r'},m_{r'})$ since $m_r=m_{r'}$ and $m^+_r=m^-_{r'}$). Therefore $\{r,r'\}$ interfere with all the other requests. Hence $I=\{r,r'\}$.

If $I$ contains exactly one unimodal request, the result follows from Lemma~\ref{lem:nospan}.

Henceforth, we may assume $I\cap R^\vee = \varnothing$.

Let $I^-\coloneqq I\cap R^-$ and $I^+\coloneqq I\cap R^+$ ; note that $I^-\cap I^+=\varnothing$.
Assume for the sake of contradiction that $I^-$ contains two requests $r_1$ and $r_2$, and $I^+$ contains two requests $r_3$ and $r_4$.
By Lemma~\ref{lem:interference}\ref{it:converging}, $t_{r_1}$ and $t_{r_2}$ are not related in $T$, and by Lemma~\ref{lem:interference}\ref{it:diverging}, $s_{r_3}$ and $s_{r_4}$ are not related in $T$. 
Moreover, by Lemma~\ref{lem:interference}\ref{it:mixte}, $s_{r_3}$ is related with $t_{r_1}$ and $t_{r_2}$, so it is a common ancestor of both. The same holds for $s_{r_4}$, which contradicts the fact that $s_{r_3}$ and $s_{r_4}$ are not related.
\end{proof}

\section{Independence and clique numbers of interference graphs}\label{sec:inde-clique}

In this section, we prove that the independence number (Section~\ref{sec:independence}) and the clique number (Section~\ref{sec:clique}) of any interference graph can be computed in polynomial time.

\subsection[alternative section title to avoid latex warnings]{Computing $\bm{\alpha(R,T)}$} \label{sec:independence}

\begin{thm}\label{thm:indep}
    Given a set $R$ of requests in a bidirected tree $T$, there is an algorithm that returns a maximum independent set of $\cI(R,T)$ in $\bigO(|R|\log |R|+|T|)$ time.
\end{thm}

\begin{proof}
    Let us root $T$ at some (arbitrary) vertex $z$, and let $\le_T$ be the ancestor relation in $T$, i.e. $x \le_T y$ for $x,y\in V(T)$ whenever $x$ is an ancestor of $y$ in $T$. As is well-known, this relation is reflexive, antisymmetric, and transitive.
    
    Let $R^+ = r_1, \ldots, r_n$ be the diverging requests in $R$ ordered with a postfix depth-first search (DFS) according to their emission arcs (so, if $s^+_{r_i}$ is an ancestor of $s^+_{r_j}$ in $T$, we must have $i\ge j$). 
    Let $S^+ \subseteq V(T)$ be the set of nodes $x\in V(T)$ such that there is a diverging request $r\in R^+$ with $x=s^+_r$. 

    \begin{claim}
    \label{claim:mis-R+}
        Let $I \subseteq R^+$ be an independent set of diverging requests. If there is a request $r\in I$ and a request $r'\in R^+\setminus I$ such that $s^+_r$ is an ancestor of $s^+_{r'}$ in $T$, then $(I\setminus \{r\})\cup \{r'\}$ is also an independent set of diverging requests.
    \end{claim}

    \begin{subproof}
        By transitivity of $\le_T$ and \Cref{lem:interference}\ref{it:diverging}, if there is a request $r'' \in R^+$ that interferes with $r'$, then $r''$ also interferes with $r$. Hence $r'' \notin I\setminus \{r\}$, from which we infer that $(I\setminus \{r\})\cup \{r'\}$ is indeed an independent set.
    \end{subproof}

    
    By \Cref{claim:mis-R+}, one naturally obtains a maximum independent set $I^+ \subseteq R^+$ of diverging requests by taking one request $r\in R^+$ such that $s^+_r = x_0$ for each minimal element (with respect to $\le_T$) in $S^+$. 

    We construct a maximum independent set $I^- \subseteq R^-$ of converging requests with a symmetric approach, i.e. we let $R^- = q_1, \ldots, q_m$ be the converging requests in $R$ ordered with a postfix depth-first search (DFS) according to their reception arcs (so, if $t^-_{r_i}$ is an ancestor of $t^-_{r_j}$ in $T$, we must have $i\ge j$), and $S^- \subseteq V(T)$ be the set of nodes $x\in V(T)$ such that there is a converging request $r\in R^-$ with $x=t^-_r$.
    
    So far, we have $|I^+|=\alpha(R^+,T)$ and $|I^-|=\alpha(R^-,T)$.

    \begin{claim}
    \label{claim:extend-is}
        Let $I \subseteq R^+$ be an independent set of diverging requests of size at least $2$, and let $S \coloneqq \sset{t_r}{r\in I}$. Let $x_0 \coloneqq \lca(S)$ be the least common ancestor of $S$, and let $r \in R^-\cup R^\vee$. Then $I\cup \{r\}$ is an independent set if and only if $s_r$ is an ancestor of $x_0$ in $T$. 
    \end{claim}

    \begin{subproof}
        If $s_r$ is an ancestor of $x_0$ in $T$, the fact that $I\cup \{r\}$ is an independent set follows from \Cref{lem:interference}\ref{it:mixte} if $r\in R^-$ or from  \Cref{lem:interference}\ref{it:uni-div} if $r\in R^\vee$. 

        Conversely, if $s_r$ is not an ancestor of $x_0$ in $T$, then there is a request $r'\in I$ such that $s_r$ and $t_{r'}$ are not related in $T$, hence $r$ and $r'$ interfere, so  $I\cup \{r\}$ is not an independent set.
    \end{subproof}
    
    If 
    there is a request $r^- \in R^-\cup R^\vee$ such that $s_{r^-}$ is an ancestor of $\lca(\sset{t_r}{r\in I^+})$ in $T$,  then we add it to $I^+$.
    We proceed symetrically for $I^-$; if
    there is a request $r^+ \in R^+\cup R^\vee$ such that $t_{r^+}$ is an ancestor of $\lca(\sset{s_r}{r\in I^-})$ in $T$, then we add it to $I^-$.
    
    We let $I$ be the larger of $I^-, I^+$. If $|I|\ge 2$, then we return $I$, otherwise we look for a pair of requests $r,r' \in R$ such that $r$ and $r'$ do not interfere and return $I \coloneqq \{r,r'\}$ if we find one. 
    If nothing has been returned so far, we conclude that $\alpha(R)=1$ (assuming $R\neq\emptyset$), and we may return $I \coloneqq \{r\}$ for any request $r\in R$. 

    \begin{claim}
        The returned independent set $I$ has size $\alpha(R,T)$.
    \end{claim}

    \begin{subproof}If $\alpha(R,T) \le 2$ or $\alpha(R,T) = \max \big\{ \alpha(R^+,T), \alpha(R^-,T) \big\}$ it is clear that the above algorithm returns a maximum independent set of $\cI(R,T)$. 
    Let us now assume that $\alpha(R,T) > \max \big\{2,  \alpha(R^+,T), \alpha(R^-,T) \big\}$, and let $\wI$ be a maximum independent set of $\cI(R,T)$. 
    From \Cref{lem:colour-size-1}, we infer that $\wI$ either contains $|\wI|-1$ diverging requests, or $|\wI|-1$ converging requests. 
    By symmetry, let us assume that $\wI$ contains $|\wI|-1$ diverging requests, which implies that  
    \[ \alpha(R,T) = 1 + \alpha(R^+,T), \]
    and so in particular $\alpha(R^+,T)\ge 2$.

    If $r^-$ does not exist, then by \Cref{claim:extend-is} there is no maximum independent set of $\cI(R^+,T)$ which can be extended into a larger independent set of $\cI(R,T)$, a contradiction ($\wI \cap R^+$ is such an independent set). We conclude that $r^-$ exists, so at the end of the algorithm $|I^+|=1+\alpha(R^+,T) = \alpha(R,T)$. Hence the returned independent set is maximum in $\cI(R,T)$. 
    The case $|\wI \cap R^-| = |\wI|-1$ is obtained by symmetry.        
    \end{subproof}

    \begin{claim}
        The above algorithm can be implemented so that it runs in $\bigO(|R|\log |R|+|T|)$ time.
    \end{claim}

    \begin{subproof}
        On can construct $S^+$ and $S^-$ in $\bigO(|R|+|T|)$ time by performing a DFS of $T$, and then initialize $I^+$ and $I^-$ in $\bigO(|R|)$ time. 
        It is standard that $\lca$ queries can be done in constant time in $T$ after a $\bigO(|T|)$ time preprocessing, and so computing $\lca(\sset{t_r}{r\in I^+})$ can be done in $\bigO(|R|+|T|)$ time. 

        We are done unless $|I^+|=|I^-|=1$ at the end of the algorithm, in which case we need to look for a pair $\{r,r'\}$ of independent requests efficiently.
        This can trivially be done in $\bigO(|R|^2)$ time by testing all pairs of requests. In the following we show how to perform this task in $\bigO(|R|\log |R| + |T|)$ time.
        To do so, we will need to perform a lot of ancestor queries in $T$; to make those efficiently, we translate them into range queries as follows.
        Let $E =(v_1, \ldots, v_m)$ be an Eulerian tour of the bidirected tree $T$, starting (and terminating) at the root $z$. From $E$, one can construct two arrays $\tin$ and $\tout$ such that, for each $x\in V(T)$, $\tin[x] = \min \sset{i}{v_i = x}$ and $\tout[x] = \max \sset{j}{v_j = x}$.
        Observe that $(v_{\tin[x]}, \ldots , v_{\tout[x]})$ is an Eulerian tour of the bidirected subtree $T[x]$ of $T$ rooted in $x$. In particular, $x$ is an ancestor of $y$ in $T$ if and only if $\tin[y] \in [\tin[x], \tout[x]]$ (which can be checked in constant time).

        There are several possible combinations for the types of requests $r$ and $r'$ among unimodal, converging, or diverging, and each of them needs to be handled separately because it has its own interference condition (described in \Cref{lem:interference}).
        \paragraph{Unimodal/unimodal pairs} To detect such an independent pair, we begin by partitioning $R^\vee$ into equivalence classes for the relation $\sim$ defined by $r \sim r'$ if $m_r = m_{r'}$. This can be performed in (expected) time $\bigO(|R|)$ by relying on a hashtable whose keys are $\sset{m_r}{r\in R^\vee}$ (or in deterministic time $\bigO(|R|\log |R|)$ with a dictionary relying on a balanced binary search tree). 
        Now, in each equivalence class $X$, we seek for a pair $r,r'\in X$ such that $s_r$ and $t_{r'}$ as well as $t_r$ and $s_{r'}$ are related in $T$. Up to symmetry and up to switching the roles of $r$ and $r'$, there are two cases to consider.
        \begin{enumerate}[label=(\roman*)]
        \item \label{it:containment} We first seek for a pair $r,r'\in X$ such that $s_r \le_T t_{r'}$ ($s_r$ is an ancestor of $t_{r'}$ in $T$) and $t_r \le_T s_{r'}$. To do so, for each $r\in X$, we seek a request $r'\in X\setminus \{r\}$ such that 
        \[ \begin{cases}
        \tin[t_{r'}] \in [\tin[s_r],\tout[s_r]]; \mbox{and}\\
        \tin[s_{r'}] \in [\tin[t_r],\tout[t_r]].
        \end{cases}\]
        We are therefore performing $|X|$ $2$-dimensional range queries. Using a $2$-dimensional range tree substructure, this can be done in $\bigO(|X|\log |X|)$ time (see \cite[Theorem 5.11]{BCKO08}).
        
        Note that, by switching the roles of $r$ and $r'$, this also covers the pairs $(r,r')$ with $s_{r'} \le_T t_r$ and $t_{r'} \le_T s_r$.
        
        \item \label{it:overlap} We then seek for a pair $r,r'\in X$ such that $s_r \le_T t_{r'}$ and $s_{r'} \le_T t_r$.
        Leveraging on \Cref{cor:conv-comp}\ref{it:conv-comp:cobipartite}, $\cI(X,T)$ is a cobipartite graph, whose parts $A,B$ can be computed in linear time. We can therefore restrict ourselves to pairs $(r,r') \in A\times B$ and $(r,r')\in B\times A$ (the latter are treated similary by symmetry).
        Let us respectively sort the requests $r\in A$ and $r'\in B$ with a prefix DFS ordering with respect to $s_r$ and $t_{r'}$ in $T$ (equivalently, we sort them with increasing $\tin[s_r]$ and $\tin[t_{r'}]$). To find an independent pair $(r,r')\in A\times B$, we proceed as follows.

        \begin{itemize}
            \item Let $D$ be an empty $1$-dimensional range tree (i.e. a balanced binary search tree). At every step of the algorithm, $D$ will contain a subset of requests $r\in A$ whose node $s_r$ all lie on a common branch of $T$, each associated with the key $\tin[t_r]$. 
            \item 
            For each $i\in \{1, \ldots, |E|\}$ (recall that $E$ is the Eulerian tour of $T$): 
            \begin{itemize}
                \item add to $D$ all requests $r\in A$ such that $\tin[s_r]=i$ with the key $\tin[t_r]$;
                \item remove from $D$ all requests $r\in A$ such that $\tout[s_r]=i-1$;
                \item For each request $r'\in B$ such that $\tin[t_{r'}]=i$, test in time $\bigO(\log |D|)$ whether there is a request $r\in D$ with $\tin[t_r] \in [\tin[s_{r'}], \tout[s_{r'}]]$. If so, this means that $s_{r'} \le_T t_r$, and because $r\in D$ we have $\tin[t_{r'}]=i \in [\tin[s_r],\tout[s_r]]$, i.e. $s_r\le_T t_{r'}$. Hence we can stop the search and return the pair $(r,r')$.
            \end{itemize}
        \end{itemize}
        
        The above procedure runs in $\bigO(|B|\log |A| + |T|)$ time, and its symmetric (which switches the roles of $A$ and $B$) runs in $\bigO(|A|\log |B| + |T|)$ time.


        Note that seeking for a pair $(r,r')$ such that $t_r \le_T s_{r'}$ and $t_{r'} \le_T s_r$ is symmetric (by switching the roles of $s_r$ and $t_r$) and can therefore be done similarly.
        \end{enumerate}

        \paragraph{Unimodal/diverging pairs} For each unimodal request $r\in R^\vee$ we seek for a request $r'\in R^+$ such that $m_r \le_T s_{r'}$ and either $s_r \le_T t_{r'}$ or $t_{r'}\le_T s_r$. The first condition translates to 
            \[ \begin{cases}
            \tin[s_{r'}] \in [\tin[m_r],\tout[m_r]]; \mbox{and}\\
            \tin[t_{r'}] \in [\tin[s_r],\tout[s_r]],
            \end{cases}\]
            and is treated similarly to case \ref{it:containment} of the unimodal/unimodal pairs. 
            The second translates to 
            \[ \begin{cases}
            \tin[s_{r'}] \in [\tin[m_r],\tout[m_r]]; \mbox{and}\\
            \tin[s_r] \in [\tin[t_{r'}],\tout[t_{r'}]],
            \end{cases}\]
            is treated similarly to case \ref{it:overlap} of the unimodal/unimodal pairs. 
            
            \paragraph{Unimodal/converging pairs} This is symmetric to the previous case, and is treated similarly.
            
            \paragraph{Diverging/converging pairs}  For each $r\in R^+$, we seek a request $r'\in R^-$ such that $t_r \le_T s_{r'}$ (i.e. $\tin[s_{r'}] \in [\tin[t_r],\tout[t_r]]$), and for each $r\in R^-$, we seek a request $r'\in R^+$ such that $s_r \le_T t_{r'}$ (i.e. $\tin[t_{r'}] \in [\tin[s_r],\tout[s_r]]$). This can be done with range queries over a $1$-dimensional range tree, for a total cost of $\bigO(|R|\log |R|)$.

        \bigskip
        The overall cost of all these tests is in $\bigO(|R|\log |R| + |T|)$ time.
    \end{subproof}
    Combining the above claims, there is a $\bigO(|R|\log|R|+|T|)$ time implementation of an algorithm that returns a maximum independent set of $\cI(R,T)$ when given $(R,T)$ as an input, as desired.
\end{proof}


    

\subsection[alternative section title to avoid latex warnings]{Computing $\bm{\omega(R,T)}$} \label{sec:clique}

\begin{thm}
\label{thm:clique}
Given a set $R$ of requests on a bidirected tree $T$, one can compute the clique number $\omega(R,T)$ of $\cI(R,T)$ in $\bigO(|R|^{4.5})$ time.
\end{thm}

In order to prove Theorem~\ref{thm:clique}, we need to establish a structural result on the cliques in an interference graph.
We begin with the following fact.

\begin{lemma}
\label{lem:branch}
Let $T$ be a rooted bidirected tree, and let $Q$ be a branch of $T$.
Let $r_0$ be a request that intersects $Q$, and let $r$ be a unimodal request arc-disjoint from $Q$. Then $r$ and $r_0$ interfere.
\end{lemma}

\begin{proof} Free to consider the converses of $r$ and $r_0$, which interfere if and only if $r$ and $r_0$ interfere, we may assume that $r_0$ intersects $Q$ in a converging arc $(u,v)$. 

Assume for a contradiction that $r$ and $r_0$ do not interfere. Then, by Lemma~\ref{lem:interference}~(v), $m_r$ is an ancestor of $v$, and $t_r$ and $u$ are related. Hence $(m_r, m^+_r)$ is an arc in $Q$, a contradiction.
\end{proof}

\begin{lemma}
\label{lem:clique}
Let $R$ be a set of requests on a bidirected tree $T$, and let $W$ be a clique in the interference graph $\cI(R,T)$.
There exist a bough $Q_W$ of $T$ and a vertex $x_W\in V(Q_W)$  such that all requests in $W$ that are arc-disjoint from $Q_W$ are unimodal in $T[x_W]$.

Moreover, the set of pairs $(Q_W, x_W)$ over all cliques $W$ of $\cI(R,T)$ has only a quadratic size (in $|R|$).
\end{lemma}

\begin{proof}
Let $x_0\in V(T)$ be chosen arbitrarily, and let $T_0\coloneqq T[x_0]$.
Let $W^+$, $W^-$, and $W^\vee$ be the set of diverging, converging, and unimodal requests of $W$ in $T_0$, respectively.
By Corollary~\ref{lem:clique-branch}, there is a branch $Q^+$ of $T_0$ such that the emission arc of every request in $W^+$ is contained in $Q^+$, and there is a branch $Q^-$ of $T_0$ such that  the reception arc of every request in $W^+$ is contained in $Q^-$. 
So the result holds when $W^-=\varnothing$ by taking $x_W=x_0$ and $Q_W$ any bough containing $Q^+$, and when $W^+=\varnothing$ by taking $x_W=x_0$ and $Q_W$ any bough containing $Q^-$.


We may now assume that both $W^+$ and $W^-$ are non-empty. Let $x_W$ be the ultimate vertex of $Q^-\cap Q^+$, that is the vertex $x_W$ such that $T\langle x_0, x_W\rangle=Q^-\cap Q^+$. 
Let $y^-$ (resp. $y^-$) be the leaf of $Q^-$ (resp. $Q^+$), and let $Q_W=T\langle y^-, y^+\rangle$. 

If all diverging requests have their emission arc in $Q^+ \cap Q^-$, we can take $x_W = x_0$ and $Q_W$ any bough containing $Q^-$, and the result holds because all requests in $W^+$ and $W^-$ intersect $Q_W$ and all of the other are in $W^\vee$ and are thus unimodal.
Similarly, we get the result if all converging requests have their reception arc in $Q^- \cap Q^+$.
Henceforth we may assume that there are a diverging request $r^+$ with emission arc in $T\langle x_W, y^+\rangle$ and a converging request $r^-$ with reception arc in $T\langle x_W, y^-\rangle$.


Let $r$ be a request in $W$ that is arc-disjoint from $Q_W$.\\
Assume that $r$ is converging (resp. diverging) in $T_0$. Then its reception (resp. emission) arc must be in $Q^+ \cap Q^-$ or it would not be arc-disjoint from $Q_W$. Moreover, $r$ cannot be diverging (resp. converging) in $T[x_W]$, for otherwise it would not interfere with $r^+$ (resp. $r^-$) and so not be in $W$. Thus $r$ is unimodal in $T[x_W]$.\\
Assume now that $r$ is unimodal in $T_0$. 
Suppose for a contradiction that $r$ is converging in $T[x_W]$. Then the ultimate arc of $r$ must belong to $T(x_0,x_W)$. Hence $r$ does not interfere with $r^-$, a contradiction.
Similarly, we get a contradiction if $r$ is diverging in $T[x_W]$.
Thus $r$ is unimodal in $T[x_W]$.

We have proved that all requests in $W$ arc-disjoint from $Q_W$ are unimodal in $T[x_W]$.

\medskip
Each bough is entirely defined by its two extremities; there are at most as many possible $Q_W$ as pairs of leaves in $\cI(R, T)$. On top of that, if $T$ is always rooted into the same vertex $x_0$, there is a unique choice for $x_W$ given $Q_W$ (this is the vertex of $Q_W$ closest to $x_0$), so the number of pairs $(Q_W, x_W)$ is only quadratic (in $|R|$).
\end{proof}

We are now ready to prove Theorem~\ref{thm:clique}.

\begin{proof}[Proof of Theorem~\ref{thm:clique}]
Let $x_0\in V(T)$ be chosen arbitrarily, and let $T_0\coloneqq T[x_0]$.


For every leaf $y_i \in V(T_0)$, let $Q_i$ be the branch from $x_0$ to $y_i$. For every pair $(i,j)$, we let $Q_{ij}\coloneqq T\langle y_i,y_j \rangle$ if $i\neq j$, and $Q_{ii}\coloneqq Q_i$.
We let $x_{ij}$ be the deepest vertex from $V(Q_i)\cap V(Q_j)$ in $T_0$. 
Let $R_{ij}$ be the set of requests from $R$ that intersect $Q_{ij}$, and let $R'_{ij}$ be the set of unimodal requests from $R$ in $T[x_{ij}]$ that are arc-disjoint from $Q_{ij}$.
Observe that each of the graphs $\cI(R_{ij},T)$ and $\cI(R'_{ij},T)$ is cobipartite. Hence, by K\H onig's Theorem, computing its clique number boils down to solving the maximum matching problem in the complement. This can be done in $\bigO(|R|^{5/2})$ using an algorithm of Hopcroft and Karp~\cite{HoKa1973}.

Let $W$ be a maximum clique of $\cI(R,T)$. 
By Lemma~\ref{lem:clique}, there exists a pair $\{i,j\}$ such that $W \subseteq R_{ij} \cup R'_{ij}$.
So one has $\omega(R,T) \le \omega(R_{ij},T) + \omega(R'_{ij},T)$. On the other hand, by Lemma~\ref{lem:branch}, every pair of requests $(r,r') \in R_{ij}\times R'_{ij}$ interfere, so the union of a clique in $\cI(R_{ij},T)$ and a clique in $\cI(R'_{ij},T)$ forms a clique in $\cI(R,T)$. As a conclusion,
\[ \omega(R,T) = \max_{i,j} \Big( \omega(R_{ij},T) + \omega(R'_{ij},T)\Big).\]

Since $T$ is a possibly reduced tree and so $|V(T)|\leq 2 |R| +2$ (see Subsection~\ref{subsec:reducing}), there are $\bigO(|R|^2)$ possible choices for $\{i,j\}$, so computing this maximum can be done in $\bigO(|R|^{4.5})$ time.
\end{proof}

\section[alternative section title to avoid latex warnings]{Approximating the chromatic number and $\bm{\chi}$-boundedness  of interference graphs}\label{sec:chi-bound}

From previous results, we can deduce the following bound on the chromatic number  $\chi(R,T)$ of $\cI(R,T)$.

\begin{lemma} \label{lem:chi-bound-3}
    Let $R$ be a set of requests on a bidirected tree $T$. We have  
\[ \chi(R,T) \le \omega(R^+,T)+\omega(R^-,T)+\omega(R^{\vee},T) \le 2\,\chi(R,T).\]
\end{lemma}
\begin{proof}
Since $R=R^-\cup R^+ \cup R^\vee$, we have  $\chi(R,T) \le \chi(R^-,T) + \chi(R^+,T) + \chi(R^\vee,T)$.
By Corollary~\ref{cor:conv-comp}, $\chi(R^-,T)=\omega(R^-,T)$, $\chi(R^+,T)=\omega(R^+,T)$, and $\chi(R^\vee,T)=\omega(R^\vee,T)$. Therefore, 
\begin{eqnarray}
\chi(R,T) & \le &  \omega(R^+,T)+\omega(R^-,T)+\omega(R^{\vee},T).\label{ineg}
\end{eqnarray}

Now, let $c$ be a $\chi(R,T)$-colouring of the set of requests $R$ on $T$. Let $k^+$, $k^-$, and $k^\vee$ be the number of colours of $c$ that appear on $R^+$, $R^-$, and $R^\vee$, respectively. By Lemma~\ref{lem:nospan}, a given colour cannot appear simultaneously on $R^+$, $R^-$, and $R^\vee$, hence $k^+ + k^- + k^\vee \le 2 \chi(R,T)$.
On the other hand, $k^-\ge \chi(R^-,T) = \omega(R^-,T)$, and similarly $k^+\ge \omega(R^+,T)$ and $k^\vee \ge \omega(R^{\vee},T)$.
Hence 
\begin{eqnarray}
\omega(R^+,T)+\omega(R^-,T)+\omega(R^{\vee},T) & \leq & 2\, \chi(R,T).\label{ineg2}
\end{eqnarray}
\end{proof}

\begin{corollary}\label{cor:approximate}
  There is a $\bigO(|R|^{5/2})$-time $2$-approximation algorithm for the chromatic number of a given interference graph   $\cI(R,T)$.   
\end{corollary}
\begin{proof}
Pick a root $x$ for $T$. Colour $R^-$ with 
$\omega(R^-,T)$ colours, $R^+$ with $\omega(R^+,T)$ other colours, and $R^\vee$ with $\omega(R^{\vee},T)$ colours.
This is possible in $\bigO(|R|^{5/2})$ time, by using the algorithm of Hopcroft and Karp~\cite{HoKa1973} to colour $\cI(R^\vee, T)$ (which is a cobipartite graph), and the algorithm of Hoàng~\cite{Hoang94} to colour $\cI(R^-,T)$ and $\cI(R^+,T)$ (which are comparability graphs).   

We then obtain a colouring of  $\cI(R,T)$ with $\omega(R^+,T)+\omega(R^-,T)+\omega(R^{\vee},T) \leq  2\chi(R,T)$ colours.
\end{proof}

Clearly, each of $\omega(R^+,T)$, $\omega(R^-,T)$, and $\omega(R^{\vee},T)$ are at most $\omega(R,T)$. Thus Lemma~\ref{lem:chi-bound-3} yields $\chi(R,T) \le 3\,\omega(R,T)$. 
We now prove a better bound on 
$\chi(R,T)$ in terms of $\omega(R,T)$.

\begin{thm}\label{thm:chi-bound}
Let $R$ be a set of requests on a bidirected tree $T$. Then 
\[ \chi(R,T) \le 2\,\omega(R,T).\]
\end{thm}

A {\it nice pair} $(R,T)$ is a set of requests in a tree such that
each request has length at least $2$. 
By subdividing each pair of opposite arcs (that is replacing the arcs $(u,v)$ and $(v,u)$ by the arcs $(u,w)$, $(w,v)$, $(v,w)$, $(w,u)$, where $w=w(\{u,v\})$ is a new vertex), one easily gets the following.

\begin{lemma}\label{lem:nice}
Let $R$ be a set of requests in a tree $T$.
There exists a nice pair $(\tilde{R}, \tilde{T})$ such that
${\mathcal I}(R, T)={\mathcal I}(\tilde{R}, \tilde{T})$.
\end{lemma}

\begin{lemma}
\label{lem:root}
Let $(R,T)$ be a nice pair. Then there exists a node $x\in V(T)$ such that 
\[ \omega(R^+,T[x])+\omega(R^-,T[x]) \le \omega(R,T).\]
\end{lemma}

\begin{proof}
A directed path {\it meets} a request if it contains its emission arc or its reception arc.
Let $Q=(q_0,q_1, \dots ,q_\ell)$ be a directed path of $T$ that meets the maximum number $m$ of requests of $R$. Wihtout loss of generality, we may assume that $q_0$ and $q_{\ell}$ are leaves of $T$. Note that the requests meeting $Q$ form a clique, hence $m \le \omega(R,T)$.
Let $T_0\coloneqq T[q_0]$. 
Let $Q'=(q'_{\ell'},q'_{\ell'-1}, \dots , q'_0)$ be a directed path in $T_0$ with terminal vertex $q'_0=q_0$ that meets the maximum number $p$ of requests.
By definition, $p\leq m$.

Let $i\ge 0$ be the maximum index such that $q'_i=q_i$.

For every $j\le i$, let $M^-_j$ (resp. $M^+_j$, $M^\pm_j$) be the set of requests that are met by $Q[q_0,q_j]$  but not $Q[q_j,q_\ell]$  (resp. $Q[q_j,q_\ell]$ but not  $Q[q_0,q_j]$, both $Q[q_0,q_j]$ and $Q[q_j,q_\ell]$).
Set $m^-_j=|M^-_j|$, $m^+_j=|M^+_j|$ and  $m^\pm_j=|M^\pm_j|$.
By definition, $m^-_j+m^+_j+m^\pm_j=m$.

Similarly, for every $j\le i$, let $P^-_j$ (resp. $P^+_j$, $P^\pm_j$) be the set of requests that are met by $Q'[q'_\ell,q_j]$  but not $Q'[q_j,q_0]$  (resp. $Q'[q_j,q_0]$ but not  $Q'[q'_\ell,q_j]$, both $Q'[q'_\ell,q_j]$ and $Q'[q_j,q_0]$).
Set $p^-_j=|P^-_j|$, $p^+_j=|P^+_j|$ and  $p^\pm_j=|P^\pm_j|$. We have
$p^-_j+p^+_j+p^\pm_j=p$. 

Observe that $m^+_0+p^-_0 = m+p$. Moreover $m^+_i+p^-_i$ is less than 
the number of requests met by $Q'[q_\ell',q'_i] \cup Q[q_i, q_\ell]=T[q'_\ell, q_\ell]$ which form a clique.
Hence $m^+_i+p^-_i \leq \omega(R,T)$.
Let $s$ be the smallest integer such that $m^+_s+p^-_s \leq \omega(R,T)$.

\begin{claim}\label{claim:m+p}
$m^-_s+p^+_s \leq \omega(R,T)$.
\end{claim}
\begin{subproof}
If $s=0$, then we have the result because $m^-_0+p^+_0=0$ by definition.
Assume now that $s>0$. 
Since $(R,T)$ is a nice pair, each request has length at least $2$ in $T$. Thus a request cannot be in $M^+_{s-1}$ and in $M^-_{s}$.  
Therefore $m^-_{s} \le m^-_{s-1}+m^\pm_{s-1} = m- m^+_{s-1}$.
Similarly, $p^+_s \leq p- p^-_{s-1}$.
Hence $m^-_s+p^+_s \leq m+p - m^+_{s-1} - p^-_{s-1}$.
But $p\leq m \leq \omega(R,T)$ and, by definition of $s$,  $m^+_{s-1}+p^-_{s-1} \geq \omega(R,T)$.
Thus $m^-_s+p^+_s \leq \omega(R,T)$.
\end{subproof}


Let $T_s \coloneqq T[q_s]$. Let us show that
\[ \omega(R^-,T_s) + \omega(R^+,T_s) \le \omega(R,T). \]

Let $W^-$ be a set of converging requests in $T_s$ that forms a maximum clique of size $\omega(R^-,T_s)$. By Lemma~\ref{lem:clique-branch}~(i), the reception arc of every request of $W^-$ is included in a common in-branch $P^-$ of $T_s$. Observe that either $P^-$ is arc-disjoint from $Q[q_0,q_s]$, or it is arc-disjoint from $Q[q_\ell, q_s]$. In the first case, the union of $P^-$ and $Q[q_s,q_0]$ forms a directed path that meets $\omega(R^-,T_s)+p^+_s+p^\pm_s$ requests. By definition, $p^-_s+p^+_s+p^\pm_s = p$, 
so $\omega(R^-,T_s) \leq p^-_s$.
In the second case, the union of $P^-$ and $Q[q_s,q_\ell]$ forms a directed path that meets $\omega(R^-,T_s)+m^+_s+m^\pm$ requests. By definition,
$m^-_s+m^+_s+m^\pm = m$, so $\omega(R^-,T_s) \leq m^-_s$.
In both cases, we have $\omega(R^-,T_s) \le \max(m^-_s,p^-_s)$. By symmetry, we also have $\omega(R^+,T_s)\le \max(m^+_s,p^+_s)$. So
\[ \omega(R^-,T_s)+\omega(R^+,T_s) \le \max \{m^-_s+m^+_s, m^-_s+p^+_s, p^-_s+m^+_s, p^-_s+p^+_s\} \le \omega(R,T). \qedhere \]
\end{proof}

Lemmas~\ref{lem:nice}, \ref{lem:root}  and Inequality~\eqref{ineg} directly imply Theorem~\ref{thm:chi-bound}.

\section[alternative section title to avoid latex warnings]{Solving {\sc Interference $\bm{k}$-Colourability} in polynomial time} \label{sec:kcolourability}

\subsection[alternative section title to avoid latex warnings]{{\sc Interference 3-Colourability}} \label{sec:3colourability}

In this subsection, we show a polynomial-time algorithm solving
{\sc Interference 3-Colourability}.
We need some preliminaries.

A \emph{$k$-list assignment} $L$ of a graph $G$ is a mapping that assigns to every vertex
$v$ a set of $k$ non-negative integers, called \emph{colours}. An \emph{$L$-colouring}
of $G$ is a mapping $c:V\to\mathbb{N}$ such that
$c(v)\in L(v)$ for every $v\in V$.
\textsc{ $k$-List Colourability} consists in, given graph $G$ and a $k$-list assignment $L$, deciding whether there exists an $L$-colouring of $G$.

The following result is well-known and can be easily proved directly or using a reduction to $2$-SAT.
\begin{proposition}\label{prop:2listcol}
   $2$-List Colourability can be solved in $\bigO(|V(G)| + |E(G)|)$-time.
\end{proposition}

\begin{lemma}
\label{lem:dominating}
Let $R$ be a set of requests on a bidirected tree $T$. Then either $\cI(R,T)$ contains a dominating set of size at most $3$, or is a comparability graph.
\end{lemma}
\begin{proof}
Since we are considering reduced trees as described in Subsection~\ref{subsec:reducing}, all the leaves of $T$ are contained in a request. 

Let $s_1$ be a leaf and $r_1$ a request containing it.
By directional duality, we may assume that $s_1$ is the first vertex of $r_1$.

Let $A^+_1$ be the set of arcs in $T^+[s_1]$ that are contained in a request of $R$. It is not empty because it contains the arcs of $r_1$.  
Let $a_2=t_2^-t_2$ be an arc in  $A^+_1$ which is the furthest from $s_1$ in $T$. Then there is no arc of $A^+_1$ whose head is a descendant of $t_2$ in $T[s_1]$.
Let $r_2$ be a request containing $a_2$. (Possibly $r_2=r_1$.)

Let $T_2$ be the bidirected subtree of $T$ induced by the descendant of $t_2$ in $T[s_1]$.
Let $R'$ be the set of requests interfering with none of $r_1,r_2$.
Note that all requests of $R'$ are contained in $T[t_2,s_1]\cup T^-_2[t_2]$.

If a request $r$ contains an arc in $T[t_2,s_1]$, then it interferes with all requests of $R'$, so $\{r_1, r_2, r\}$ is a dominating set of $\cI(R,T)$. Henceforth, we may assume that all requests of $R'$ are contained in $T^-_2[t_2]$.

If a request $r$ contains a diverging arc in $T[t_2]$, then by definition of $t_2$, this arc is not in  $T_2$. Therefore every request of $R'$ interferes on $r$ and $\{r_1, r_2, r\}$ is a dominating set of $\cI(R,T)$.
Henceforth, we may assume that all requests are converging in $T[t_2]$.
But then, by Corollary~\ref{cor:conv-comp}\ref{it:conv-comp:comparability-minus}, $\cI(R,T)$ is a comparability graph.
\end{proof}

\begin{remark}\label{rem:dom}
The above proof can easily be translated into a $\bigO(|R|)$-time algorithm that either finds a dominating set of size at most $3$ or returns that $\cI(R,T)$ is a comparability graph.
\end{remark}

Lemma~\ref{lem:dominating} can be used to solve \textsc{Interference 3-Colourability} in quadratic time.

\begin{corollary}
\textsc{Interference 3-Colourability} can be solved in $\bigO(|R|^2)$ time.
\end{corollary}

\begin{proof}
Let $R$ be a set of requests in a bidirected tree $T$, and set $G=\cI(R,T)$.
Following Remark~\ref{rem:dom}, in linear time, either we find a dominating set $S$ of size at most $3$ in $G$, or we get that $G$ is a comparability graph.

In the second case, one can compute $\chi(G)$ in $\bigO(|V(G)|^2)= \bigO(|R|^2)$ time using the algorithm of Hoàng~\cite{Hoang94}.

In the first case, we enumerate the $3^{|S|}\leq 27$ possible
colourings of $S$, and, for each of them, we check whether this partial colouring extends to $G$. This can be done in $\bigO(|V(G)| + |E(G)|)= \bigO(|R|^2)$ time using \textsc{$2$-List-Colourability} (Proposition~\ref{prop:2listcol}). Indeed the set of available colours at each vertex of $V(G)\setminus S$ has size at most $2$ because at least one colour is forbidden by a neighbour in $S$.
\end{proof}

We have just proved that \textsc{Interference $3$-Colourability} is polynomial-time solvable, and would like to extend that result to \textsc{Interference $k$-Colourability} for larger values of $k$. Unfortunately, the above method cannot be used in that latter setting.
Indeed, we use the fact that \textsc{ $2$-List Colourability} is polynomial-time solvable. But as we shall now see, \textsc{$3$-List Colourability} is NP-complete even on interference graphs.

The following theorem is certainly well-known, but we include the short proof for completeness.

\begin{thm}
\label{thm:3-list-color-NPC}
\textsc{$3$-List Colourability} is NP-complete on complete bipartite graphs.
\end{thm}

\begin{proof}
Reduction from \textsc{$3$-SAT}.


Let $\mathcal{C}=(C_1, \ldots, C_m)$ be an instance of \textsc{$3$-SAT} on a set $X$ of $n$ variables $x_1, \ldots, x_n$. 
We denote $C_j = \alpha_j \vee \beta_j \vee \gamma_j$.
Let $(U,V)$ be the bipartition of $K_{m,n}$ with $U=\{u_1,\ldots,u_m\}$ and $V=\{v_1,\ldots,v_n\}$. 
Let $L$ be the list assignment defined by $L(u_j)=\{\alpha_j,\beta_j,\gamma_j\}$ for all $j\in [m]$, and $L(v_i)=\{x_i,\overline{x_i}\}$ for all $i\in [n]$.

If $\mathcal{C}$ is satisfiable, let $\phi\colon X \to \{\text{true}, \text{false}\}$ be a truth assignment satisfying $\mathcal{C}$. We let $\psi$ be the $L$-colouring such that $\psi(u_j)$ is any literal of $C_j$ that is true, and $\psi(v_i)$ is the literal of $x_i$ that is false. Then $\psi$ is a proper $L$-colouring of $G$.

The converse is straightforward.

If one wants all lists of $L$ to have length $3$, then one might add an additional universal colour $\delta$ to each list $L(v_i)$, add an additional set $U_0$ of $6$ vertices to $U$ with lists $\{\delta\}\cup S$ for every $S\in \binom{[4]}{2}$, and an additional set $V_0$ of  $3$ vertices to $V$ with all possible lists in $\binom{[4]}{3}$. Observe that at least two distinct colours $x,y\in [4]$ must appear in any proper $L$-colouring of $G\langle V_0\rangle$, which forces the vertex of $U_0$ with list $\{x,y,\delta\}$ to be coloured $\delta$. Hence the colour $\delta$ is forbidden from all the lists of $V$, and we are back in the situation described above.
\end{proof}

As we shall now see, there is a set of requests whose interference graph is the complete bipartite graph $K_{m,n}$ (for every $m,n\ge 1$). So from \Cref{thm:3-list-color-NPC} we infer the following.

\begin{corollary}
\textsc{$3$-List Colourability} is NP-complete on the class of interference graphs $\cI(R,T)$, even when the tree $T$ is a star and all requests of $R$ have length $1$.
\end{corollary}

\begin{proof}
Let $T=K_{1,m+n}$ be rooted in its unique internal node $x_0$.
Let $(L_1,L_2)$ be a partition of the leaves of $T$ with $L_1=m$ and $L_2=n$.
Let $R$ be composed of $m$ converging requests of the form $y \to x_0$ for every $y\in L_1$, and of $n$ diverging requests of the form $x_0 \to y$ for every $y \in L_2$. Then $\cI(R,T)$ is the complete bipartite graph $K_{m,n}$ with bipartition $(L_1, L_2)$.
\end{proof}

\subsection[alternative section title to avoid latex warnings]{{\bfseries{\sc Interference $\bm{k}$-Colourability}} when  $\bm{k\geq 4}$} \label{sec:4colourability}

\begin{thm} \label{th:FPT}
 \textsc{Interference $k$-Colourability} can be solved in time
 $\bigO(k36^k|R|^3)$.
\end{thm}

\begin{proof}
Let $R$ be a set of requests on a bidirected tree $T$ rooted at $z$.
Set $G=\cI(R,T)$. 

Let $\phi$ be a proper $k$-colouring of $G$. By Lemma~\ref{lem:colour-size-1}, each colour $i\in [k]$ can be of three types as follows:
\begin{enumerate}[label=(\roman*)]
    \item $i$ is {\it mainly converging} if at most one request from $R^+\cup R^\vee$ is coloured $i$;
    \item $i$ is {\it mainly diverging} if at most one request from $R^-\cup R^\vee$ is coloured $i$;
    \item $i$ is {\it unimodal} if its colour class consists of exactly two requests from $R^\vee$.
\end{enumerate}

If a non-diverging (respectively non-converging) request is coloured with a mainly diverging (respectively converging) colour, we say that this is an \emph{exceptional} request.
Note that a colour class could be both mainly converging and mainly diverging, in which case it is of the form $\{r,r'\}$ where $r$ is a converging request $r'$ is a diverging request; if $t_r$ is closer to the root $z$ than $s_{r'}$ then $r$ is the exceptional request of that colour class, otherwise $r'$ is.

We first show that we may assume that the exceptional requests in $\phi$ can be covered with two cliques.

\begin{claim} \label{claim1}
There exists a proper $k$-colouring of $G$ such that all non-diverging exceptional requests form a clique and all non-converging exceptional requests form a clique.
\end{claim}
\begin{subproof}
Consider a proper colouring $\phi$ that minimizes the number of exceptional requests.
Assume for a contradiction that there exist two independent non-diverging exceptional requests $r,r'$. Set $\gamma=\phi(r)$ and $\gamma' = \phi(r')$. Note that $\gamma\neq \gamma'$ by definition of exceptional request.
The diverging requests coloured $\gamma$ are contained in $T^+_{s_r} \cup T[m_r,s_r]$ (recall that $m_r=t_r$ if $r$ is converging), where $T^+_{s_r}$ is the subtree of $T^+$ rooted at $s_r$; in particular, for every diverging request $r_0$ coloured gamma, $s_{r_0}^+$ is a descendent of $m_r^-$ in $T$.
Likewise, for every diverging request $r_0'$ coloured $\gamma'$, $s_{r'_0}^+$ is a descendent of $m_{r'}^-$ in $T$.

Since $r$ and $r'$ are non-adjacent, $m_r^-$ and $m_{r'}^-$ are incomparable in $T$; so are in particular $x$ and $x'$ for every descendent $x$ of $m_r^-$ and $x'$ of $m_{r'}^-$. Hence the diverging requests coloured with $\gamma$ or $\gamma'$ form together an independent set.
Therefore we may recolour all of them with colour $\gamma$ and recolour $r$ (and $r'$) with colour $\gamma'$. Doing so, we obtain a proper $k$-colouring of $G$ with two fewer exceptional requests than $\phi$, a contradiction.

By directional duality, we get the result for non-converging exceptional requests. 
\end{subproof}

Let $\phi$ be a proper $k$-colouring of $G$ where the set $Q^+$ of non-converging exceptional requests is a clique, as well as the set $Q^-$ of non-diverging exceptional requests, and let $Q \coloneqq Q^+ \cup Q^-$ be the set of exceptional requests in $\phi$. 
We denote by $\phi_0$ the partial colouring $\restrict{\phi}{Q}$ induced by $\phi$ on the subgraph $G\langle Q\rangle$. Let $\ext$ be the set of proper $k$-colourings $\phi'$ of $G$ such that $\restrict{\phi'}{Q} = \phi_0$, $Q^+$ is the set of non-converging exceptional requests of $\phi'$, and $Q^-$ its set of non-diverging exceptional requests; observe that this set is non-empty since $\phi \in \ext$.
We will show that, if we are given $\phi_0$, then we can retrieve a proper $k$-colouring of $G$ in polynomial time. 

Let $Q^- = q_1, \ldots, q_t$ be the non-diverging exceptional requests of $\phi$, ordered with a postfix depth-first search (DFS) according to their source nodes (so, if $s_{q_j}$ is an ancestor of $s_{q_i}$ in $T$, we must have $i\le j$). Free to permute the colours, we may assume that $\phi_0(q_i)=i$ for each $i\in [t]$. 
Let $D = r_1, \ldots, r_n$ be the diverging non-exceptional requests, ordered with a prefix DFS according to their emission arcs (so, if $s_{r_i}^+$ is an ancestor of $s_{r_j}^+$ in $T$, we must have $i \le j$). 
We let $\phi_1$ be the colouring obtained from $\phi_0$ by extending it to $D$ greedily in that order by using the smallest available colour for each $r_i\in D$.

\begin{claim}
\label{claim2}
    There is a proper $k$-colouring $\phi' \in \ext$ of $G$ such that $\phi_1(r) \in [t] \implies \phi'(r)=\phi_1(r)$ for every request $r \in D$.
\end{claim}

\begin{subproof}
    Given a colouring $\phi' \in \ext$, we say that a request $r \in D$ is \emph{bad in $\phi'$} if $\phi_1(r) \in [t]$ and $\phi'(r) \neq \phi_1(r)$. 
    Let us assume for the sake of contradiction that every colouring in $\ext$ contains at least one bad request, and let $\phi' \in \ext$ be the one that maximises the smallest index $i_0$ of a bad request. So, setting $D_0 \coloneqq \{r_i \in D : i < i_0\}$, no request in $D_0$ is bad. This means that, for every colour $\gamma \in [t]$, $\phi_1^{-1}(\gamma) \cap D_0 = \phi'^{-1}(\gamma)\cap D_0$.
    Let us write $r\coloneqq r_{i_0}$, $\alpha \coloneqq \phi_1(r) \in [t]$, and $\beta \coloneqq \phi'(r)$. Since $\alpha \in [t]$, there is an exceptional request $q\in Q^-$ with $\phi_1(q) = \phi_0(q)=\alpha$.
    By construction, $\alpha = \min \Big([k] -  \phi_0(N(r)\cap Q) - \phi_1(N(r)\cap D_0)\Big)$.
    Let $\phi''$ be a copy of $\phi'$, in which we perform the following recolouring steps:
    \begin{itemize}
        \item $\phi''(r) \gets \alpha$;
        \item for every $r'\in D$ such that $\phi'(r')=\alpha$ and $s^+_{r'}$ is a descendent of $s^+_{r}$ in $T$, we do $\phi''(r') \gets \beta$. 
    \end{itemize}
    We claim that $\phi''$ is a proper $k$-colouring of $G$ in which no request $r_i$ with $i\le i_0$ is bad, which contradicts the choice of $\phi'$.
    To see this, we first argue that $\alpha \notin \phi''(N(r))$.
    Indeed, $\alpha$ is mainly diverging in $\phi'$, so apart from $q$ (which is not in $N(r)$ since $\phi_1(r)=\phi_1(q)=\alpha$), all requests coloured $\alpha$ in $\phi'$ are in $D$. We have $\alpha \notin \phi_1(N(r) \cap D_0)$, and since no request in $D_0$ is bad, in particular $\alpha \notin \phi'(N(r) \cap D_0)$. We infer that $N(r) \cap \phi'^{-1}(\alpha) \subseteq \{r_i \in D : i > i_0\}$. So every neighbour $r'$ of $r$ coloured $\alpha$ in $\phi'$ is such that $s_{r'}^+$ is a descendent of $s_r^+$ in $T$, which implies that $\phi''(r')=\beta$. All the other requests have the same colour in $\phi'$ and $\phi''$, so we conclude that $\alpha\notin \phi''(N(r))$.
    Next, if there is $r'\in D$ such that $\phi'(r')=\alpha$ and $s^+_{r'}$ is a descendent of $s^+_{r}$ in $T$, we argue that $r$ is the only neighbour of $r'$ coloured $\beta$ in $\phi'$. 
    Observe that $\beta > \alpha$; indeed either $\beta \notin [t]$ and this is obvious, or else $\phi'^{-1}(\beta)\cap D_0 = \phi_1^{-1}(\beta)\cap D_0$, so $\beta \notin \phi_1(N(r) \cap D_0)$, hence by construction $\phi_1(r)\le \beta$.
    Since $r$ is not exceptional, there is at most one non-diverging request $q'$ coloured $\beta$, in which case $q'\in Q^-$. Since $\beta > \alpha$, we infer that $s_{q'}$ is not a descendant of $s_q$ in $T$, and since $\phi'(r)=\phi'(q')$, we infer that $r$ and $q'$ do not interfere, so $s_{q'}$ is an ancestor of $t_r$ in $T$, from which we infer (because $r$ and $q$ do not interfere) that $s_{q'}$ is an ancestor of $s_q$ in $T$. In particular, $q'$ and $r'$ do not interfere.
    All other requests $r''$ coloured $\beta$ are diverging, and do not interfere with $r$, so $s_{r''}^+$ and $s_r^+$ are unrelated in $T$, which implies that $s_{r''}$ and $s_{r'}$ are also unrelated in $T$. 
    We conclude that $\phi''$ is proper, as desired. 
    To end the proof, there remains to show that no request $r_i$ with $i\le i_0$ is bad in $\phi''$. No request from $D_0$ is bad in $\phi'$, and $\alpha \notin \phi_1(N(r) \cap D_0)$, so no request from $D_0$ is recoloured in $\phi''$, i.e. $\restrict{\phi'}{D_0} = \restrict{\phi''}{D_0}$. Hence no request from $D_0$ is bad in $\phi''$; neither is $r=r_{i_0}$, which ends the proof.
\end{subproof}

By directional duality, a similar statement holds for the converging requests. 

\medskip

The algorithm begins with the following preprocessing, which can be easily performed in time $\bigO(|R|^2)$. 
It first checks some obvious conditions of non $k$-colourability: firstly, it checks whether there is a directed branch intersecting more than $k$ requests
which then form a clique of size greater than $k$; secondly, it checks whether $|R^\vee|  > 2k$ in which case $\cI(R^\vee,T)$ contains a clique of size greater than $k$ because it is a cobipartite graph by Corollary~\ref{cor:conv-comp}~\ref{it:conv-comp:cobipartite}. In the affirmative, the graph $G$ is not $k$-colourable, and the algorithm returns "No". 
Otherwise, in the negative, no directed branch intersects more than $k$ requests, and $|R^\vee|\le 2k$.
Then the algorithm orders the diverging requests $r\in R^   +$ with a prefix DFS ordering according to the node $s_r^+$, and for each of them records the list of diverging requests $r'\in N(r) \cap R^+$ that precede $r$ in that order into a binding \textsf{ancestor[r]} (this list has size less than $k$).
It proceeds similarly for the set $R^-$ of converging requests.

Next, the algorithm searches for the $k$-colouring of $G$ promised by Claim~\ref{claim2}. It proceeds by enumerating all possible choices for $(Q^-,Q^+)$ as follows.
It first considers all unimodal requests and, for each of them, decides whether it should belong to $Q^-$, $Q^+$, or neither. Since $|R^\vee|\le 2k$, there are at most $3^{2k}$ choices. 
For each of them, it first checks whether the complement of $G\langle R^\vee \setminus (Q^-\cup Q^+)\rangle$ has a perfect matching $M$; this can be done in time $\bigO(k^{2.5})$ by using the algorithm of Hopcroft and Karp~\cite{HoKa1973}. If not, it deduces that this specific choice of $(Q^- \cap R^\vee,Q^+ \cap R^\vee)$ is invalid and considers the next one.
Otherwise, there are at most $2^k|R|$ possibilities for the converging requests of $Q^-$ (at most $|R|$ choices for the branch, and $2^k$ choices for a subset of converging requests on that branch), and at most $2^k|R|$ possibilities for the diverging requests of $Q^+$ (by symmetry).
For each possible choice of $(Q^-,Q^+)$, it orders $Q^- = q_1, \ldots, q_t$ with a postfix DFS according to the source nodes, and $Q^+ = q_{t+1}, \ldots, q_{t+t'}$ with a postfix DFS according to the reception nodes. It then fixes $\phi_0(q_i) \coloneqq i$ for each $i\in [t+t']$, and colours each matched pair of unimodal non-exceptional requests with a distinct colour in $\{t+t'+1, \ldots, t+t'+|M|\}$. 
By Claim~\ref{claim2} and its directional dual, we may extend $\phi_0$ to $\phi_1$ greedily by using only colours from $[t]$ for the set $R^+ \setminus Q^+$ of non-exceptional diverging requests (if no colour from $[t]$ is available for a given request $r\in R^+\setminus Q^+$, we leave it uncoloured), and only colours from $\{t+1, \ldots t+t'\}$ for the set $R^-$ of non-exceptional converging requests (again, leaving some of them uncoloured if needed).
All remaining uncoloured requests of $R^+$ are coloured greedily with a new set of $k^+$ colours, by following the preprocessed order on $R^+$. This returns an optimal colouring of these requests, in time $\bigO(k|R|)$ (for each request $r$, it suffices to compute in time $\bigO(k)$ the set of colours assigned to \textsf{ancestor[r]} to decide which colour to assign to $r$).
Similarly, the remaining uncoloured requests of $R^-$ are coloured greedily with a new set of $k^-$ colours.
The resulting colouring uses $t+t'+|M|+k^++k^-$ colours. If this number is less or equal to $k$, then the algorithm stops and returns it; otherwise, it considers the next possible choices for $(Q^-,Q^+)$.

Since it performs $\bigO(36^k|R|^2)$ extension tests taking $\bigO(k|R|)$ time each, the complexity of the algorithm is 
$\bigO(k36^k|R|^3)$. 
\end{proof}

\section{Conclusion and further research} \label{sec:conclusion}

In Theorem~\ref{thm:clique}, we showed a
$\bigO(|R|^{4.5})$-time algorithm that finds the clique number of an interference graph. It is very likely that a faster algorithm can be found.

\begin{problem}
Find an algorithm that computes the clique number of a given interference graph on $|R|$ vertices in less than $\bigO(|R|^{4.5})$ time.
\end{problem}

Regarding the chromatic number of interference graphs, we showed a 2-approximation $\bigO(|R|^{5/2})$-time algorithm for $\chi(R,T)$ (Corollary~\ref{cor:approximate}).
Then, we have shown that the \textsc{Interference Colourability} problem of deciding whether $\chi(R,T) \leq k$ is FPT when parameterized on $k$. To do so, we proposed an algorithms that solve \textsc{Interference $3$-Colourability} in time $\bigO(|R|^2)$ and \textsc{Interference $k$-Colourability} in time $\bigO(k36^k|R|^3)$ when $k\geq 4$.
However, the complexity of \textsc{Interference Colourability} remains unknown.

\begin{problem}
Is \textsc{Interference Colourability} polynomial-time solvable ? Is it NP-complete ?
\end{problem}

Another question is how good it can be approximated in polynomial time. In Corollary~\ref{cor:approximate}, we showed that it can be $2$-approximated but it might be better approximated.

\begin{problem}
Find a polynomial-time $r$-approximation algorithm for some $r<2$ or prove that none exists (up to standard hypotheses) for some $r\geq 1$. 
\end{problem}

Our $2$-approximation exploits the fact that the class of interference graphs admits a $\chi$-bounding function, that is  a function $f$ such that $\chi(G) \leq f(\omega(G))$ for every interference graph $G$. Theorem~\ref{thm:chi-bound} shows that $f\colon t \mapsto 2t$ is a $\chi$-bounding function of interference graphs.

\begin{problem}
What is the smallest $\chi$-bounding function for interference graphs, that is the smallest function $f^*$ such that $\chi(R,T) \leq f^*(\omega(R,T))$ for every set $R$ of requests in a bidirected tree $T$ ?
\end{problem}

We note that $f^*(2t) \ge \ceil{5t/2}$.
Indeed consider the bidirected tree with vertex set $\{a,b,c,d\}$ and edge set $\{ab, ba, bc, cb, bd, db\}$.
Let $R_1$ (resp. $R_2$, $R_3$, $R_4$, $R_5$)  be a set of $t$ requests equal to $(a,b)$ (resp. $(b,c)$, $(d,b)$, $(b,a)$, $(c,b,d)$), and let $R=R_1\cup R_2 \cup R_3 \cup R_4 \cup R_5$.
One can easily check that ${\cal I}(R,T)$ is $C_5[K_t]$, the lexicographic product of the $5$-cycle $C_5$ by the complete graph $K_t$.
Hence $\omega(R,T) = 2t$ and $\chi(R,T)=\ceil{5t/2}$.



\medskip

A natural way to find a proper colouring of a graph is to successively select a maximum independent set.
This is a polynomial-time algorithm for interference graph because a maximum independent set can be computed in polynomial time by Theorem~\ref{thm:indep}.

\begin{problem}
    Does the algorithm consisting in finding successively a maximum independent set in the subgraph induced by the uncoloured vertices have an approximation guarantee on interference graphs ?
    That is, does there exist a constant $C$ such that this algorithm produces a colouring with at most $C \cdot \chi(G)$ for every interference graph $G$ ?
\end{problem}

Note that the set $R$ of requests in the bidirecteed tree $T$ depicted in Figure~\ref{fig:3/2} shows that such a constant $C$ is at least $3/2$.
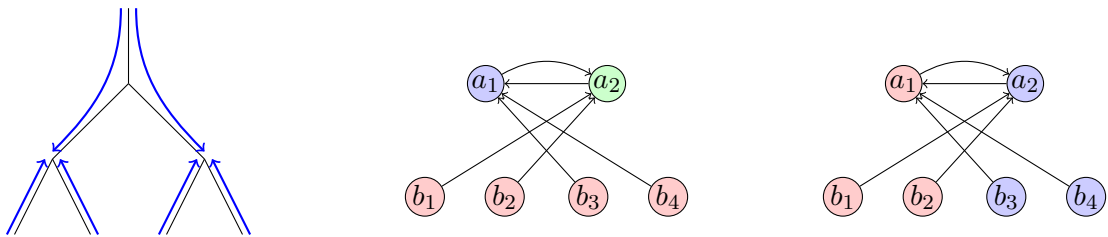
\begin{figure}[!ht]
\begin{center}
\begin{tikzpicture}
\tikzstyle{vertex}=[circle,draw, minimum size=12pt, scale=1, inner sep=0.5pt]
\coordinate (x) at (0,0);
\coordinate (x') at (-0.1,0);
\coordinate (x'') at (0.1,0);
\coordinate (y) at (0,-1);
\coordinate (z0) at (-1,-2);
\coordinate (z0') at (-1.1,-2);
\coordinate (z0'') at (-0.9,-2);
\coordinate (z0''') at (-1,-1.9);
\coordinate (z1) at (1,-2);
\coordinate (z1') at (0.9,-2);
\coordinate (z1'') at (1.1,-2);
\coordinate (z1''') at (1,-1.9);
\coordinate (t0) at (-1.5,-3);
\coordinate (t0') at (-1.6,-3);
\coordinate (t0'') at (-1.4,-3);
\coordinate (t1) at (-0.5,-3);
\coordinate (t1') at (-0.6,-3);
\coordinate (t1'') at (-0.4,-3);
\coordinate (t2) at (0.5,-3);
\coordinate (t2') at (0.4,-3);
\coordinate (t2'') at (0.6,-3);
\coordinate (t3) at (1.5,-3);
\coordinate (t3') at (1.4,-3);
\coordinate (t3'') at (1.6,-3);

\draw (x) -- (y);
\draw (y) -- (z0);
\draw (y) -- (z1);
\draw (z0) -- (t0);
\draw (z0) -- (t1);
\draw (z1) -- (t2);
\draw (z1) -- (t3);

\draw[thick, blue,->] (x') to[out=270, in=45] (z0''');
\draw[thick, blue,->] (x'') to[out=270, in=135] (z1''');
\draw[thick, blue,->] (t0') to (z0');
\draw[thick, blue,->] (t1'') to (z0'');
\draw[thick, blue,->] (t2') to (z1');
\draw[thick, blue,->] (t3'') to (z1'');

\begin{scope}[xshift=5.5cm]
 \node[vertex, fill=blue!20] (a1) at (-0.8,-1){$a_1$};
\node[vertex, fill=green!20] (a2) at (0.8,-1){$a_2$};
 \node[vertex, fill=red!20] (b1) at (-1.6,-2.5){$b_1$};
\node[vertex, fill=red!20] (b2) at (-0.55,-2.5){$b_2$};
\node[vertex, fill=red!20] (b3) at (0.55,-2.5){$b_3$};
\node[vertex, fill=red!20] (b4) at (1.6,-2.5){$b_4$};

\draw[->] (a1) to[bend left] (a2);
\draw[->] (a2) to (a1);
\draw[->] (b4) to (a1);
\draw[->] (b3) to (a1);
\draw[->] (b1) to (a2);
\draw[->] (b2) to (a2);
\end{scope}

\begin{scope}[xshift=11cm]
 \node[vertex, fill=red!20] (a1) at (-0.8,-1){$a_1$};
\node[vertex, fill=blue!20] (a2) at (0.8,-1){$a_2$};
 \node[vertex, fill=red!20] (b1) at (-1.6,-2.5){$b_1$};
\node[vertex, fill=red!20] (b2) at (-0.55,-2.5){$b_2$};
\node[vertex, fill=blue!20] (b3) at (0.55,-2.5){$b_3$};
\node[vertex, fill=blue!20] (b4) at (1.6,-2.5){$b_4$};

\draw[->] (a1) to[bend left] (a2);
\draw[->] (a2) to (a1);
\draw[->] (b4) to (a1);
\draw[->] (b3) to (a1);
\draw[->] (b1) to (a2);
\draw[->] (b2) to (a2);
\end{scope}

\end{tikzpicture}
\caption{
A set $R$ of requests in blue  in a bidirected tree $T$ in black (left) ; the interference digraph ${\cal I}(R,T)$ coloured by the algorithm (middle) and optimally coloured (right).}\label{fig:3/2}
\end{center}
\end{figure}

\bibliographystyle{plain}
\bibliography{references}
\end{document}